\documentclass[runningheads]{llncs}

\usepackage{amsthm}
\usepackage{bm}
\usepackage{graphicx}
\usepackage{amsmath}
\usepackage{amsfonts}
\usepackage{tikz}
\usepackage{subfigure}
\usepackage{framed}
\usepackage[switch]{lineno}
\newtheorem{observation}{Observation}

\begin{document}

\title{Incentive Compatible Mechanism for Influential Agent Selection}
%
\author{Xiuzhen Zhang \and
Yao Zhang \and
Dengji Zhao}
\authorrunning{X. Zhang et al.}
\institute{ShanghaiTech University, 
Shanghai, China \\
\email{\{zhangxzh1,zhangyao1,zhaodj\}@shanghaitech.edu.cn}}
\maketitle              
\begin{abstract}

Selecting the most influential agent in a network has huge practical value in applications. However, in many scenarios, the graph structure can only be known from agents' reports on their connections. In a self-interested setting, agents may strategically hide some connections to make themselves seem to be more important. In this paper, we study the incentive compatible (IC) selection mechanism to prevent such manipulations. Specifically, we model the progeny of an agent as her influence power, i.e., the number of nodes in the subgraph rooted at her. We then propose the Geometric Mechanism, which selects an agent with at least $1/2$ of the optimal progeny in expectation under the properties of incentive compatibility and fairness. Fairness requires that two roots with the same contribution in two graphs are assigned the same probability. Furthermore, we prove an upper bound of $1/(1+\ln 2)$ for any incentive compatible and fair selection mechanisms.

\end{abstract}

\section{Introduction}
The motivation for influential agent selection in a network comes from real-world scenarios, where networks are constructed from the following/referral relationships among agents and the most influential agents are selected for various purposes (e.g., information diffusion~\cite{kimura2007extracting} or opinion aggregation~\cite{mohammadinejad2019consensus}). However, in many cases, the selected agents are rewarded (e.g., coupons or prizes), and the network structures can only be known from their reports on their following relationships. Hence, agents have incentives to strategically misreport their relationships to make themselves selected, which causes a deviation from the optimal results. An effective selection mechanism should be able to prevent such manipulations, i.e., agents cannot increase their chances to be selected by misreporting, which is a key property called incentive compatibility.

There have been many studies about incentive compatible selection mechanisms with different influential measurements for various purposes (e.g., maximizing the in-degrees of the selected agent~\cite{alon2011sum,fischer2015optimal,caragiannis2021impartial}). In this paper, we focus on selecting an agent with the largest progeny. For this purpose, the following two papers are the most related studies. Babichenko et al.~\cite{babichenko2018incentive} proposed the Two Path Mechanism based on random walks. Although their mechanism achieves a good approximation ratio of $2/3$ between the expected and the optimal influence in trees, it has no guaranteed performance in forests or general directed acyclic graphs (DAGs). Furthermore, Babichenko et al.~\cite{BabichenkoDT20} advanced these results by proposing another two selection mechanisms with an approximation ratio of about $1/3$ in forests. In these two papers, the authors assumed that agents can add their out-edges to any other agents in the network. This strong assumption limited the design of incentive compatible mechanisms. Also, in many cases, agents cannot follow someone they do not know. Therefore, we focus on the manipulation of hiding the connections they already have. In practice, it is possible that two agents know each other, but they are not connected. Then they are more than welcome to connect with each other, which is not harmful for the selection. Moreover, there still exists a noticeable gap between the approximation ratios of existing mechanisms and a known upper bound of $4/5$~\cite{BabichenkoDT20} for all incentive compatible selection mechanisms in forests. Therefore, by restricting the manipulations of agents, we want to investigate whether we can do better.

Furthermore, the previous studies mainly explored the forests, while in this paper, we also looked at DAGs. A DAG forms naturally in many applications because there exist sequential orders for agents to join the network. Each agent can only connect to others who joined the network before her, e.g., a reference or referral relationship network. Then, in a DAG, each node represents an agent, and each edge represents the following relationship between two agents. 

In this setting, the action of each agent is to report a set of her out-edges, which can only be a subset of her true out-edges. The goal is to design selection mechanisms to incentivize agents to report their true out-edge sets. Besides the incentive compatibility, we also consider another desirable property called fairness. Fairness requires that two agents with the maximum progeny in two graphs share the same probability of being selected if their progeny make no difference in both graphs (the formal definition is given in Section 2). Then, we present an incentive compatible selection mechanism with an approximation ratio of $1/2$ and prove an upper bound of $1/(1+\ln 2)$ for any incentive compatible and fair selection mechanism.

\subsection{Our Contributions}
We focus on the incentive compatible selection mechanism in DAGs.  It is natural to assign most of the probabilities to select agents with more progeny to achieve a good approximation ratio. Thus, we identify a special set of agents in each graph, called the influential set. Each agent in the set, called an influential node, is a root with the maximum progeny if deleting all her out-edges in the graph. They are actually the agents who have the chances to make themselves the optimal agent with manipulations. On the other hand, we also define a desirable property based on the graph structure, called fairness. It requires that the most influential nodes (the agents with the maximum progeny) in two graphs have the same probability to be selected if the number of nodes in the two graphs, the subgraphs constructed by the two nodes' progeny, and the influential sets are all the same.

Based on these ideas, we propose the Geometric Mechanism, which only assigns positive probabilities to the influential set. Each influential node will be assigned a selection probability related to her ranking in the influential set. We prove that the Geometric Mechanism satisfies the properties of incentive compatibility and fairness and can select an agent with her progeny no less than $1/2$ of the optimal progeny in expectation. The approximation ratio of the previous mechanisms is at most $1/\ln 16$ $(\approx 0.36)$. Under the constraints of incentive compatibility and fairness, we also give an upper bound of $1/(1+\ln 2)$ for the approximation ratio of any selection mechanism, while the previous known upper bound for any incentive compatible selection mechanism is $4/5$.

\subsection{Other Related Work}
\subsubsection{Without the constraint of incentive compatibility.} Focusing on influence maximization, Kleinberg~\cite{kleinberg2007cascading} proposed two models for describing agents' diffusion behaviours in networks, i.e., the linear threshold model and the independent cascade model. It is proved to be NP-hard to select an optimal subset of agents in these two models. Following this, there are studies on efficient algorithms to achieve bounded approximation ratios between the selected agents and the optimal ones under these two models~\cite{zhang2016identifying,huang2020efficient,ko2018efficient,morone2015influence}. 
 
In cases where only one influential agent can be selected, the most related literature also studied methods to rank agents based on their abilities to influence others in a given network, i.e., their centralities in the network. A common way is to characterize their centralities based on the structure of the network. In addition to the classic centrality measurements (e.g., closeness and betweenness~\cite{kundu2011new,pal2014centrality}) or Shapley value based characterizations~\cite{narayanam2008determining}, there are also other ranking methods in real-world applications, such as PageRank \cite{page1999pagerank} where each node is assigned a weight according to its connected edges and nodes.

\subsubsection{With the constraint of incentive compatibility.} In this setting, incentive compatible selection mechanisms are implemented in two ways: with or without monetary payments. The first kind of mechanism incentivizes agents to truthfully reveal their information by offering them payments based on their reports. For example, Narayanam et al.~\cite{narahari2011incentive} considered the influence maximization problem where the network structure is known to the planner, and each agent will be assigned a fixed positive payment based on influence probabilities they reported. With monetary incentives, there are also different mechanisms proposed to prevent agents from increasing their utilities by duplicating themselves or colluding together~\cite{emek2011mechanisms,shen2019multi,zhang2020sybil}. To achieve incentive compatible mechanisms without monetary incentives, the main idea of the existing work is to design probabilistic selection mechanisms and ensure that each agent's selection probability is independent of her report~\cite{alon2011sum,aziz2016strategyproof,fischer2015optimal}. For example, Alon et al.~\cite{alon2011sum} designed randomized selection mechanisms in the setting of approval voting, where networks are constructed from agents' reports. Our work belongs to this category.

\section{The Model}
Let $\mathcal{G}^n$ be the set of all possible directed acyclic graphs (DAGs) with $n$ nodes and $\mathcal{G} = \bigcup_{n \in \mathbb{N}^*}\mathcal{G}^n$ be the set of all directed acyclic graphs. Consider a network represented by a graph $G=(N,E) \in \mathcal{G}$, where $N = \{1,2,\cdots, n\}$ is the node set and $E$ is the edge set. Each node $i \in N$ represents an agent in the network and each edge $(i,j) \in E$ indicates that agent $i$ follows (quotes) agent $j$. Let $P_i$ be the set of agents who can reach agent $i$, i.e., for all agent $j \in P_i$, there exists at least one path from $j$ to $i$ in the network. We assume $i \in P_i$. Let $p_i = |P_i|$ be agent $i$'s progeny and $p^* = \max_{i\in N} |P_i|$ be the maximum progeny in the network.

Our objective is to select the agent with the maximum progeny. However, we do not know the underlying network and can only construct the network from the following/referral relationships declared by all agents, i.e., their out-edges. In a game-theoretical setting, agents are self-interested. If we simply choose an agent $i \in N$ with the maximum progeny, agents who directly follow agent $i$ may strategically misreport their out-edges (e.g., not follow agent $i$) to increase their chances to be selected. Therefore, in this paper, our goal is to design a selection mechanism to assign each agent a proper selection probability, such that no agent can manipulate to increase her chance to be selected and it can provide a good approximation of the expected progeny in the family of DAGs.

For each agent $i \in N$, her type is denoted by her out-edges $\theta_i= \{(i,j) \mid (i,j) \in E, j \in N\}$, which is only known to her. Let $\theta = (\theta_1,\cdots,\theta_n)$ be the type of all agents and $\theta_{-i}$ be the type of all agents expect $i$. Let $\theta_i'$ be agent $i$'s report to the mechanism and  $\theta' = (\theta_1',\cdots, \theta_n')$ be the report profile of all agents. Note that agents do not know the others except for the agents they follow in the network. Then $\theta_i' \subseteq \theta_i$ should hold for all $i \in N$, which satisfies the Nested Range Condition~\cite{green1986partially} thus guarantees the revelation principles. Thereby, we focus on direct revelation mechanism design here. Let $\Phi(\theta_i)$ be the space of all possible report profiles of agent $i$ with true type $\theta_i$, i.e., $\Phi(\theta_i) = \{\theta_i' \mid \theta_i' \subseteq \theta_i\}$. Let $\Phi(\theta)$ be the set of all possible report profiles of all agents with true type profile $\theta$.

Given $n$ agents, let $\Theta^n$ be the set of all possible type profile of $n$ agents. Given $\theta \in \Theta^n$ and a report profile $\theta' \in \Phi(\theta)$, let $G(\theta')= (N,E')$ be the graph constructed from $\theta'$, where $N = \{1,2,\cdots,n\}$ and $E' = \{(i,j) \mid i,j\in N, (i,j) \in \theta' \}$. Denote the progeny of agent $i$ in graph $G(\theta')$ by $p_i(\theta')$ and the maximum progeny in this graph by $p^*(\theta')$. We give a formal definition of a selection mechanism.

\begin{definition}
A selection mechanism $\mathcal{M}$ is a family of functions $f: \Theta^n \rightarrow [0,1]^{n}$ for all $n \in \mathbb{N}^*$. Given a set of agents $N$ and their report profile $\theta'$, the mechanism $\mathcal{M}$ will give a selection probability distribution on $N$. For each agent $i \in N$, denote her selection probability by $x_i(\theta')$. We have $x_i(\theta')\in [0,1]$ for all $i \in N$ and $\sum_{i \in N}x_i(\theta') \leq 1$.
\end{definition}

Next, we define the property of \emph{incentive compatibility} for a selection mechanism, which incentivizes agents to report their out-edges truthfully.

\begin{definition}[Incentive Compatible]
A selection mechanism $\mathcal{M}$ is \textbf{incentive compatible (IC)} if for all $N$, all $i \in N$, all $\theta \in \Theta^n$, all $\theta_{-i}' \in \Phi(\theta_{-i})$ and all $\theta_i' \in \Phi(\theta_i)$, $x_i((\theta_i,\theta_{-i}')) \geq x_i((\theta_i',\theta_{-i}'))$.
\end{definition}

An incentive compatible selection mechanism guarantees that truthfully reporting her type is a dominant strategy for all agents. An intuitive realization is a uniform mechanism where each agent gets the same selection probability. However, there exists a case where most of the probabilities are assigned to agents with low progeny, thus leading to an unbounded approximation ratio. We desire an incentive compatible selection mechanism to achieve a bounded approximation ratio for all DAGs. We call this property \emph{efficiency} and define the efficiency of a selection mechanism by its approximation ratio. 

\begin{definition}
Given a set of agents $N = \{1,2,\cdots,n\}$, their true type profile $\theta \in \Theta^n$, the performance of an incentive compatible selection mechanism in the graph $G(\theta)$ is defined by
\[R(G(\theta)) = \frac{\sum_{i \in N} x_i(\theta)p_i(\theta)}{p^*(\theta)}.\]
We say an incentive compatible selection mechanism $\mathcal{M}$ is \textbf{efficient with an approximation ratio $r$} if for all $N$, all $\theta \in \Theta^n$, $R(G(\theta)) \geq r$.
\end{definition}
This property guarantees that the worst-case ratio between the expected progeny of the selected agent and the maximum progeny is at least $r$ for all DAGs. Without the constraint of incentive compatibility, an optimal selection mechanism will always choose the agent with the maximum progeny. While in the strategic setting, an agent with enough progeny can misreport to make herself the agent with the maximum progeny. We define such an agent as \textit{an influential node}. In a DAG, there can be multiple influential nodes. Thus we define them as \textit{the influential set}, denoted by $S^{inf.}$. For example, in the graph shown in Figure~\ref{fig:inf}, when removing agent $3$'s out-edge, agent $3$ will be the root with the maximum progeny, same for agents $1$ and $2$. The formal definitions are as follows.
\begin{figure}[!htb]
  \centering
  \includegraphics[width=.45\linewidth]{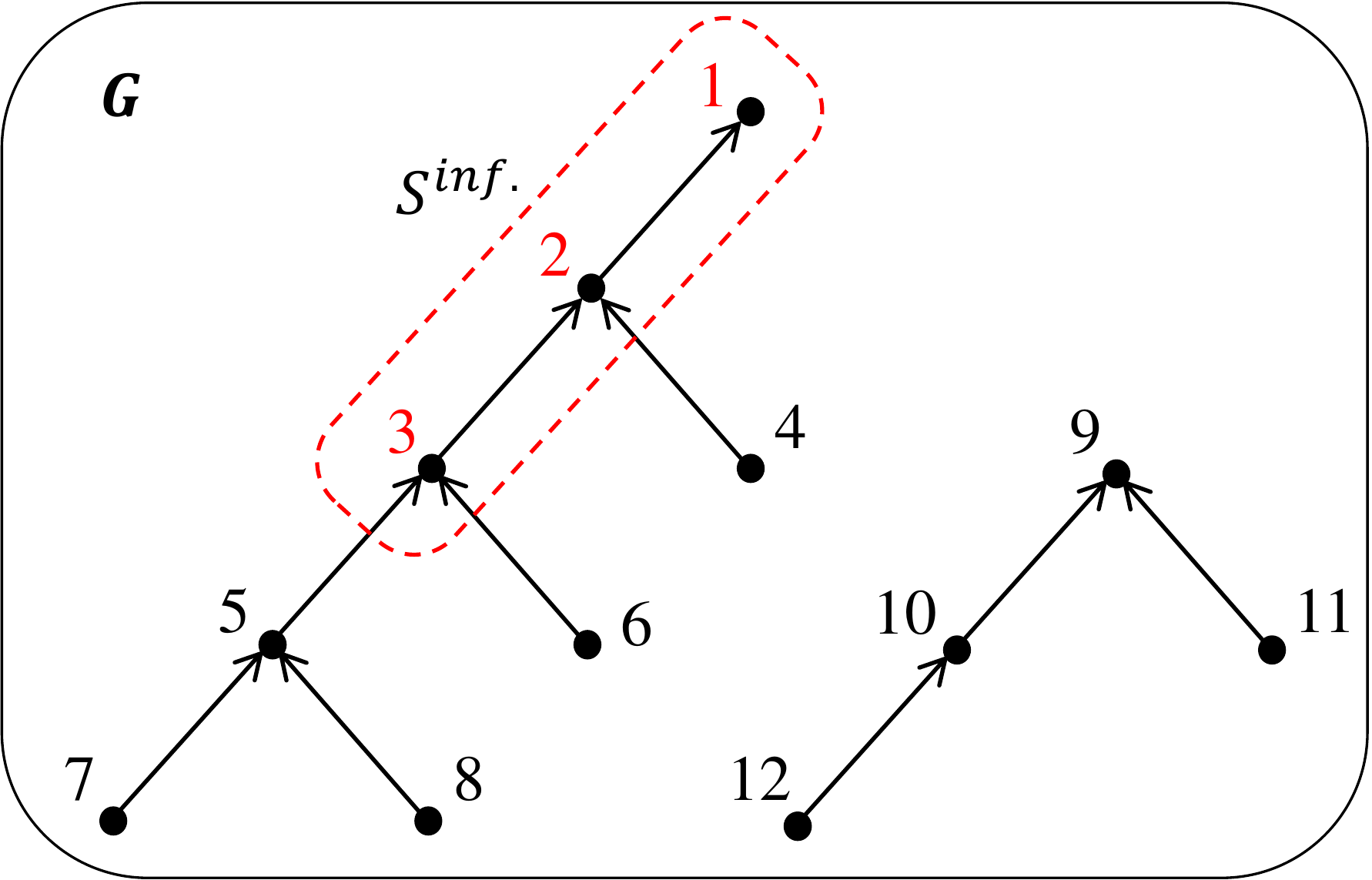}
  \caption{An example for illustrating the definition of influential nodes: agents $1,2,3$ are the influential nodes and they form the influential set in the graph $G$. }
  \label{fig:inf}
\end{figure}

\begin{definition}\label{def:inf}
For a set of agents $N=\{1,2,\cdots,n\}$, their true type profile $\theta \in \Theta^n$ and their report profile $\theta' \in \Phi(\theta)$, an agent $i \in N$ is an \textbf{influential node} in the graph $G(\theta')$ if $p_i ((\theta_{-i}',\emptyset)) \succ p_j((\theta_{-i}',\emptyset))$ for all $j \neq i \in N$, where $p_i\succ p_j$ if either $p_i > p_j$ or $p_i = p_j$ with $i<j$.
\end{definition}

\begin{definition}\label{def:infset}
For a set of agents $N=\{1,2,\cdots,n\}$, their true type profile $\theta \in \Theta^n$ and their report profile $\theta' \in \Phi(\theta)$, the \textbf{influential set} in the graph $G(\theta')$ is the set of all influential nodes, denoted by $S^{inf.}(G(\theta')) = \{s_1,\cdots,s_m\}$, where $s_i \succ s_j$ holds if and only if $p_i \succ p_j$, $s_i \succ s_{j}$ holds for all $m\geq j> i \geq 1$ and $m = |S^{inf.}(G(\theta'))|$. 
\end{definition}

According to the above definitions, we present three observations about the properties of influential nodes.
\begin{observation}\label{ob:onepath}
Given a set of agents $N = \{1,2,\cdots,n\}$, their true type $\theta \in \Theta^n$ and their report profile $\theta' \in \Phi(\theta)$, there must exist a path that passes through all agents in $S^{inf.}(G(\theta'))$ with an increasing order of their progeny.
\end{observation}
\begin{proof}
Let the influential set be $S^{inf.}(G(\theta')) = \{s_1,\cdots,s_m\}$. The statement shows that agent $s_j$ is one of the progeny of agent $s_i$ for all $1 \leq i < j\leq m$, then we can prove it by contradiction.

Assume that there doesn't exist a path passing through all agents in the influential set, then there must be an agent $j$ such that $s_j \notin P_{s_i}$ for all $1\leq i < j$. Since $s_i,s_j \in S^{inf.}(G(\theta'))$, for all $1\leq i < j$, we have
\begin{align}
    p_{s_i}((\theta'_{-s_i},\emptyset)) &\succ p_{s_j}((\theta'_{-s_i},\emptyset)), \label{1} \\
    p_{s_j}((\theta'_{-s_j},\emptyset)) &\succ p_{s_i}((\theta'_{-s_j},\emptyset)). \label{2}
\end{align}
We also have $p_{s_j}((\theta'_{-s_j},\emptyset)) = p_{s_j}((\theta'_{-s_i},\emptyset))$ and $p_{s_i}((\theta'_{-s_j},\emptyset)) = p_{s_i}((\theta'_{-s_i},\emptyset))$ since $s_j \notin P_{s_i}$ and there is no cycle in the graph. With the lexicographical tie-breaking way, the inequality \ref{1} and \ref{2} cannot hold simultaneously. Therefore, we get a contradiction.
\end{proof}

\begin{observation}\label{ob:mosti}
Given a set of agents $N = \{1,2,\cdots, n\}$, their true type $\theta \in \Theta^n$ and their report profile $\theta' \in \Phi(\theta)$, let $S^{inf.}(G(\theta'))= \{s_1,\cdots,s_m\}$ be the influential set in the graph $G(\theta')$. Then, agent $s_1$ has no out-edges and she is the one with the maximum progeny, i.e., agent $s_1$ is \textbf{the most influential node}.
\end{observation}
\begin{proof}
We prove this statement by contradiction. Assume that agent $s_1$ has at least one out-edge. Then there must exist an agent $i \in N$ such that $s_1 \in P_i$ and $p_i((\theta_{-i}',\emptyset)) \succ p_k((\theta_{-i}',\emptyset))$ for all $k \neq i$, otherwise there must exist an agent $j \in N$ such that $s_1 \notin P_j$ and  $p_j((\theta_{-i}',\emptyset)) \succ p_i((\theta_{-i}',\emptyset))$, which means that $s_1 \notin S^{inf.}(G(\theta'))$ since $p_i((\theta_{-i}',\emptyset)) \succ p_{s_1}((\theta_{-i}',\emptyset))$. Thus, such an $i$ must exist when agent $s_1$ has out-edges. Now, we must have $i \in S^{inf.}(G(\theta'))$ and $p_i \succ p_{s_1}$, which contradicts with $p_{s_1} \succ p_{j}$ for all $j \in S^{inf.}(G(\theta'))$ and $j \neq s_1$.

Then we can conclude that agent $s_1$ has no out-edges. Since $p_{s_1}((\theta_{-s_1}',\emptyset))$ $ \succ p_{k}((\theta_{-s_1}',\emptyset))$ for all $k \neq s_1$, we can get that agent $s_1$ has the maximum progeny in the graph $G(\theta')$ and she is the most influential node.
\end{proof}

\begin{observation}\label{ob:set}
Given a set of agents $N = \{1,2,\cdots,n\}$, their true type profile $\theta \in \Theta^n$, for all agent $i \in N$, all $\theta_{-i}' \in \Phi(\theta_{-i})$, if agent $i$ is not an influential node in the graph $G((\theta_{-i}',\theta_i))$, she cannot make herself an influential node by misreporting.
\end{observation}

\begin{proof}
Given other agents' report $\theta_{-i}'$, whether an agent $i$ can be an influential node depends on the relation between $p_i((\theta_{-i}',\emptyset))$ and $p_j((\theta_{-i}',\emptyset))$, rather than the out-edges reported by agent $i$.
\end{proof}

There is one additional desirable property we consider in this paper. Consider two graphs $G,G' \in \mathcal{G}^n$ illustrated in Figure~\ref{fig:consistency}, where they have the same influential set ($S^{inf.}(G) = S^{inf.}(G')$) and $s_1$ is the most influential node in both graphs. Additionally, the subgraphs constructed by agents in $P_{s_1}$ are the same in both $G$ and $G'$ (The red parts in Figure~\ref{fig:consistency}, represented by $G(s_1) = G'(s_1)$). The only difference between the two graphs lies in the edges that are not in the subgraphs constructed by agents in $P_{s_1}$ (The yellow parts in Figure~\ref{fig:consistency}).
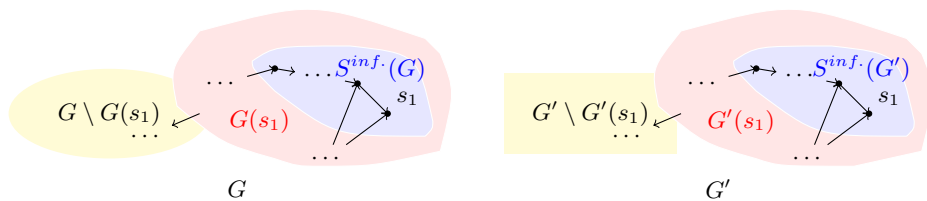
\begin{figure}[!htbp]
    \centering
    \begin{tikzpicture}
    \usetikzlibrary{shapes}
    \usetikzlibrary{decorations.pathreplacing}
    \usetikzlibrary{calc}
    \node[fill=yellow!20, ellipse, inner sep=7.5pt] (B1) at (-3.7,0) {$G\setminus G(s_1)$};
    \filldraw[draw=white, fill=red!10]
    (.7,-.5) .. controls (-.2,-.7) .. (-.9,-.7) .. controls (-1.2,-.7) .. (-1.5,-.6) .. controls (-2.2,-.4) .. (-2.8,.1) .. controls (-2.9,.6) .. (-2.5,1.1) .. controls (-1.,1.5) .. (.5,1.) .. controls (.6,.9) .. (.7,.8) .. controls (.9,.3) .. (.7,-.5);
    \filldraw[draw=white, fill=blue!10]
    (.5,-.3) .. controls (-.7,-.3) .. (-1.5,.3) .. controls (-1.8,.6) .. (-1.8,.7) .. controls (-1.5,.9) .. (-.5,.9) .. controls (.4,.8) .. (.6,-.1) .. controls (.6,-.2) .. (.5,-.3);
    \node[blue] (T1) at (-.1,.6) {$S^{inf.}(G)$};
    \node[red] (T2) at (-1.7,-.1) {$G(s_1)$};
    \filldraw (0,0) circle (1pt) node[align=center, above right] (A1) {$s_1$};
    \filldraw (-.4,.4) circle (1pt);
    \node (A2) at (-.9,.5) {$\cdots$};
    \node (A3) at (-2.2,.4) {$\cdots$};
    \node (A4) at (-.8,-.6) {$\cdots$};
    \node (A5) at (-3.2,-.3) {$\cdots$};
    \filldraw (-1.5,.6) circle (1pt);
    \draw[->] (-1.5,.6) -> (A2);
    \draw[->] (A2) -> (-.4,.4);
    \draw[->] (-.4,.4) -> (0,0);
    \draw[->] (A3) -> (-1.5,.6);
    \draw[->] (A4) -> (-.4,.4);
    \draw[->] (A4) -> (0,0);
    \draw[->] (-2.5,0) -> (A5);
    \node (G1) at (-2.0,-1.0) {$G$};
    \node[fill=yellow!20, inner sep=10.5pt] (B2) at (2.7,0) {$G'\setminus G'(s_1)$};
    \filldraw[draw=white, fill=red!10]
    (7.1,-.5) .. controls (6.2,-.7) .. (6.2,-.7) .. controls (5.2,-.7) .. (4.9,-.6) .. controls (4.2,-.4) .. (3.6,.1) .. controls (3.5,.6) .. (3.9,1.1) .. controls (5.4,1.5) .. (6.9,1.) .. controls (7.0,.9) .. (7.1,.8) .. controls (7.3,.3) .. (7.1,-.5);
    \filldraw[draw=white, fill=blue!10]
    (6.9,-.3) .. controls (5.7,-.3) .. (4.9,.3) .. controls (4.6,.6) .. (4.6,.7) .. controls (4.9,.9) .. (5.9,.9) .. controls (6.8,.8) .. (7.0,-.1) .. controls (7.0,-.2) .. (6.9,-.3);
    \node[blue] (T3) at (6.3,.6) {$S^{inf.}(G')$};
    \node[red] (T4) at (4.7,-.1) {$G'(s_1)$};
    \filldraw (6.4,0) circle (1pt) node[align=center, above right] (A6) {$s_1$};
    \filldraw (6.0,.4) circle (1pt);
    \node (A7) at (5.5,.5) {$\cdots$};
    \node (A8) at (4.2,.4) {$\cdots$};
    \node (A9) at (5.6,-.6) {$\cdots$};
    \node (A10) at (3.2,-.3) {$\cdots$};
    \filldraw (4.9,.6) circle (1pt);
    \draw[->] (4.9,.6) -> (A7);
    \draw[->] (A7) -> (6.0,.4);
    \draw[->] (6.0,.4) -> (6.4,0);
    \draw[->] (A8) -> (4.9,.6);
    \draw[->] (A9) -> (6.0,.4);
    \draw[->] (A9) -> (6.4,0);
    \draw[->] (3.9,0) -> (A10);
    \node (G2) at (4.4,-1.0) {$G'$};
    \end{tikzpicture}
    \caption{Example for fairness: in graphs $G$ and $G'$, $S^{inf.}(G) = S^{inf.}(G')$, $G(s_1) = G'(s_1)$; the only difference is in the yellow parts. Fairness requires that $x_{s_1}(G) = x_{s_1}(G')$.}
    \label{fig:consistency}
\end{figure}

We can observe that $s_1$ and all her progeny have the same contributions in the two graphs intuitively. Therefore, it is natural to require that a selection mechanism assigns the same probability to $s_1$ in the two graphs. We call this property \emph{fairness} and give the formal definition as follows.

\begin{definition}[Fairness]
For a graph $G = (N,E) \in \mathcal{G}$, define a subgraph constructed by agent $i$'s progeny as $G(i) = (P_i,E_i)$, where $E_i = \{(j,k) \mid j,k \in P_i, (j,k) \in E\}$ and $i \in N$.

A selection mechanism $\mathcal{M}$ is \textbf{fair} if for all $N$, for all $G, G' \in \mathcal{G}^n$ where $S^{inf.}(G) = S^{inf.}(G') = \{s_1,\cdots,s_m\}$ and $G(s_1) = G'(s_1)$, then $x_{s_1}(G) = x_{s_1}(G')$.
\end{definition}

\section{Geometric Mechanism}\label{mechanism}
In this section, we present the Geometric Mechanism, denoted by $\mathcal{M}_G$. In Observation~\ref{ob:set}, an agent without enough progeny cannot make herself an influential node by reducing her out-edges. Therefore, to maximize the approximation ratio, we can just assign positive selection probabilities to agents in the influential set. This is the intuition of the Geometric Mechanism.

\begin{framed}
 \noindent\textbf{Geometric Mechanism}
 
 \noindent\rule{\textwidth}{0.5pt}
 \begin{enumerate}
     \item Given the set of agents $N = \{1,2,\cdots,n\}$, their true type profile $\theta \in \Theta^n$ and their report profile $\theta' \in \Phi(\theta)$, find the influential set $S^{inf.}$ in the graph $G(\theta')$:
     \[S^{inf.}(G(\theta')) = \{ s_1,\cdots,s_m\},\]
     where $s_i \succ s_{i+1}$ for all $1 \leq i \leq m-1$.
     \item The mechanism gives the selection probability distribution on all agents as the following.
     \begin{align*}
         x_i = \begin{cases}
         1/(2^{m-j+1}), & i = s_j,\\
         0, & i \notin S^{inf.}(G(\theta')).
         \end{cases}
     \end{align*}
 \end{enumerate}
\end{framed}
Note that the Geometric Mechanism assigns each influential node a selection probability related to her ranking in the influential set. Besides, an agent' probability is decreasing when her progeny is increasing. This is reasonable because if an influential node $j$ is one of the progeny of another influential node $i$, the contribution of agent $i$ partially relies on $j$. To guarantee efficiency and incentive compatibility simultaneously, we assign a higher probability to agent $j$ compared to agent $i$. We give an example to illustrate how our mechanism works below.
\begin{example}
Consider the network $G$ shown in Figure~\ref{fig:example}.
\begin{figure}[!htb]
  \centering
  \includegraphics[width=.5\linewidth]{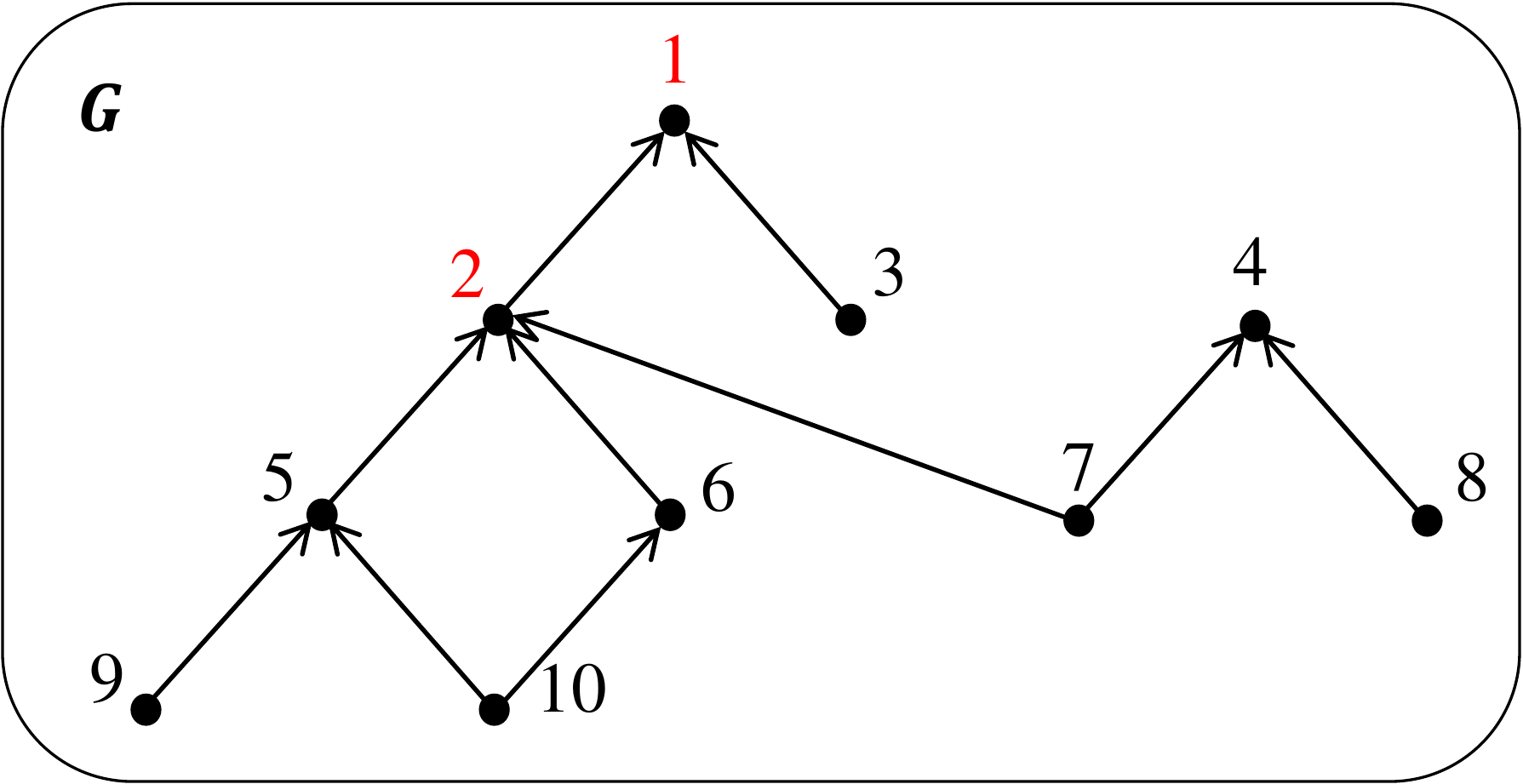}
  \caption{An example for the Geometric Mechanism.}
  \label{fig:example}
\end{figure}
We can observe that only agents $1$ and $2$ will have the largest progeny in the graphs when they have no out-edges respectively. Thus, the influential set is $S^{inf.}(G)=\{1,2\}$. Since $p_1\succ p_2$, then according to the probability assignment defined in the Geometric Mechanism, we choose agent $1$ with probability $1/4$, choose agent $2$ with probability $1/2$ and choose no agent with probability $1/4$. The expected progeny chosen by the Geometric Mechanism in this graph is
\[ \mathbb{E}[p] = \frac{1}{2}\times 6 + \frac{1}{4} \times 8 = 5. \]

On the other hand, the largest progeny is given by agent $1$, which is $8$, so that the expected ratio of the Geometric Mechanism in this graph is $5/8$.

\end{example}

Next, in Theorems~\ref{thm:GM p} and~\ref{thm:GM up}, we show that our mechanism satisfies the properties of incentive compatibility and fairness and has an approximation ratio of $1/2$ in the family of DAGs.

\begin{theorem}\label{thm:GM p}
In the family of DAGs, the Geometric Mechanism satisfies incentive compatibility and fairness.
\end{theorem}

\begin{proof}
In the following, we give the proof for these properties separately.

\noindent
\textbf{Incentive Compatibility.} Given a set of agents $N = \{1,2,\cdots,n\}$, their true type $\theta \in \Theta^n$ and their report profile $\theta' \in \Phi(\theta)$, let $G(\theta')$ be the graph constructed by $\theta'$, and $S^{inf.}(G(\theta'))$ be the influential set in $G(\theta')$. To prove that the mechanism is incentive compatible, we need to show that $x_i((\theta_{-i}',\theta_i)) \geq x_i((\theta_{-i}',\theta_i'))$ holds for all agent $i \in N$.
\begin{itemize}
    \item According to Observation~\ref{ob:set}, for agent $i \notin S^{inf.}(G((\theta_{-i}',\theta_i)))$, she cannot misreport to make herself be an influential node. Thus, her selection probability will always be zero.
    \item  If agent $i \in S^{inf.}(G((\theta_{-i}',\theta_i)))$, she cannot misreport to make herself be out of the influential set. Suppose $ S^{inf.}(G((\theta_{-i}',\theta_i))) = \{s_1,\cdots,s_m\}$ and $i = s_l$, $1\leq l\leq m$. Denote the set of influential nodes in her progeny when she truthfully reports by $S_i((\theta_{-i}',\theta_i)) = \{j \in S^{inf.}(G((\theta_{-i}',\theta_i))) \mid  p_i ((\theta_{-i}',\theta_i))$ $ \succ p_j((\theta_{-i}',\theta_i)) \}$. Then agent $i$'s selection probability in the graph $G((\theta_{-i}',\theta_i))$ is $x_i((\theta_{-i}',\theta_i)) = 1/({2^{m-l+1}}) = 1/(2^{|S_i((\theta_{-i}',\theta_i))|+1})$.
    
    When she misreports her type as $\theta_i' \subset \theta_i$, i.e., deleting a nonempty subset of her real out-edges, $p_j((\theta_{-j}',\emptyset)) \succ p_k((\theta_{-j}',\emptyset))$ still holds for all $j \in S_i((\theta_{-i}',\theta_i))$, all $k \in N$ and $k\neq j$. This can be inferred from Observation~\ref{ob:onepath}, agent $j$ is one of the progeny of agent $i$ for all $j \in S_i$. Thus, agent $i$'s report will not change agent $j$'s progeny. Moreover, some other agent $t \in P_i$ may become an influential node in the graph $G((\theta_{-i}',\theta_i'))$, since $\max_{k\in N, k\neq t}p_k((\theta_{-t}',\emptyset))$ may be decreased and $p_t((\theta_{-t}',\emptyset))$ keeps unchanged. Then we can get $S_i((\theta_{-i}',\theta_i)) \subseteq S_i((\theta_{-i}',\theta_i'))$, which implies that $x_i((\theta_{-i}',\theta_i)) = 1/{2^{|S_i((\theta_{-i}',\theta_i))|+1}} \geq x_i((\theta_{-i}',\theta_i'))= 1/{2^{|S_i((\theta_{-i}',\theta_i'))|+1}}$.
\end{itemize}

Thus, no agent can increase her probability by misreporting her type and the Geometric Mechanism satisfies incentive compatibility.

\noindent
\textbf{Fairness.} For any two graph $G,G' \in \mathcal{G}^n$, if their influential sets and the subgraphs constructed by the progeny of the most influential node are both the same, i.e., $S^{inf.}(G) = S^{inf.}(G') = \{s_1,\cdots,s_m\}$ and $G(s_1) = G'(s_1)$, according to the definition of Geometric Mechanism, agent $s_1$ will always get a selection probability of $1/2^m$. Therefore, the Geometric Mechanism satisfies fairness.
\end{proof}

\begin{theorem}\label{thm:GM up}
In the family of DAGs, the Geometric Mechanism can achieve an approximation ratio of $1/2$.
\end{theorem}

\begin{proof}
Given a graph $G = (N,E) \in \mathcal{G}$ and its influential set $S^{inf.}(G) = \{s_1,\cdots,s_m\}$, the maximum progeny is $p^* = p_{s_1}$. Then the expected ratio should be
\begin{align*}
    R &= \frac{\mathbb{E}[p]}{p^*}
    = \frac{\sum_{i \in S^{inf.}(G) }x_i p_i}{p^*} \\
    &= \frac{\sum_{i=1 }^m 1/(2^{m-i+1}) \cdot p_{s_i}}{p^*} \\
    &=  \sum_{i=2}^{m} \frac{1}{2^{m-i+1}} \cdot \frac{p_{s_i}}{p_{s_1}} + \frac{1}{2^m} \cdot \frac{p_{s_1}}{p_{s_1}} \\
    &\geq \sum_{j = 1}^{m-1} \frac{1}{2^j} \cdot \frac{1}{2} + \frac{1}{2^m} \\
    &= \frac{1}{2} - \frac{1}{2^m} + \frac{1}{2^m} = \frac{1}{2}.
\end{align*}
The inequality holds since $p_{s_i}/ p_{s_1} \geq \frac{1}{2}$ holds for all $1\leq i \leq m-1$. This can be inferred from Observation~\ref{ob:onepath}, agent $s_i$ is one of agent $s_1$'s progeny for all $i > 1$. If $p_{s_i}/ p_{s_1} < \frac{1}{2}$, then we will have $p_{s_i}((\theta_{-{s_i}},\emptyset)) \prec p_{s_1}((\theta_{-s_i},\emptyset))$, which contradicts with that $s_i \in S^{inf.}(G)$.

The expected ratio holds for any directed acyclic graph, which means that
\[r_{\mathcal{M}_G} = \min_{G \in \mathcal{G}} R(G) = \frac{1}{2}.\]
Thus we complete the proof.
\end{proof}

\section{Upper Bound and Related Discussions}\label{charater}
In this section, we further give an upper bound for any incentive compatible and fair selection mechanisms in Theorem~\ref{theorem:ub}. After that, we consider a special class of selection mechanisms, called root mechanisms (detailed in Section~\ref{openproblems}), which contains the Geometric Mechanism. Then, we propose two conjectures on whether root mechanisms and fairness will limit the upper bound of the approximation ratio.

\subsection{Upper Bound}\label{ub}
We prove an upper bound for any IC and fair selection mechanisms as below.
\begin{theorem}\label{theorem:ub}
For any incentive compatible and fair selection mechanism $\mathcal{M}$, $r_{\mathcal{M}} \leq \frac{1}{1+ \ln 2}$.
\end{theorem}

\begin{proof}
Consider the graph $G = (N,E)$ shown in Figure~\ref{fig:upperbound}, the influential set in $G$ is $S^{inf.}(G) = \{{2k-1}, {2k-2}, \cdots, k\}$. When $k \rightarrow \infty$, for each agent $i$, $i\leq k-1$, their contributions can be ignored, it is without loss of generality to assume that they get a probability of zero, i.e., $x_{i}(G)=0$. Then, applying a generic incentive compatible and fair mechanism $\mathcal{M}$ in the graph $G$, assume that $x_{i}(G) = \beta_{i-k}$ is the selection probability of agent $i$, $ \beta_{i-k} \in [0,1]$ and $ \sum_{i = k}^{2k-1} \beta_{i-k} \leq 1$.

For each agent $i \in N$, set $N_i = P_i(G)$, $N_{-i} = N\setminus N_i$, $E_i = \{(j,k) \mid j,k \in I_i, (j,k) \in E\}$ and $E_{-i} = E\setminus \{E_i \cup \theta_i\}$. Define a set of graphs $\mathcal{G}_i = \{G'= (G(i);G(-i)) \mid G(-i) = (N_{-i},E_{-i}'), E_{-i}' \subseteq E_{-i}\}$. Then for any graph $G' \in \mathcal{G}_i$, it is generated by deleting agent $i$'s out-edge and a subset of out-edges of agent $i$'s parent nodes, illustrated in Figure~\ref{fig:upperbound}. For any $i \geq k$ and any graph $G' \in \mathcal{G}_i$, the influential set in the graph $G'$ should be $S^{inf.}(G') = \{i,{i-1},\cdots,k\}$. 

\begin{figure}[!htb]
  \centering
  \includegraphics[width=.6\linewidth]{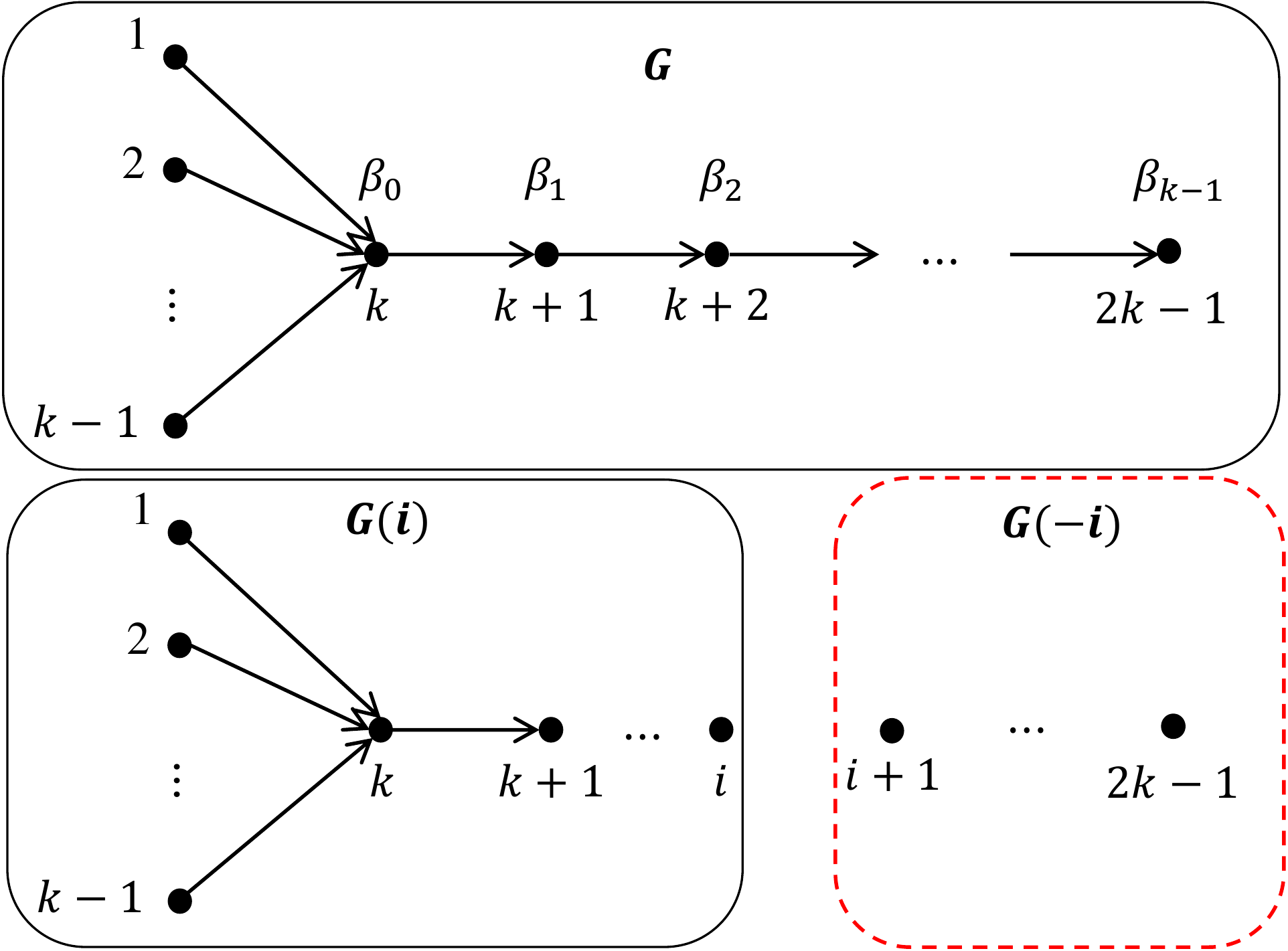}
  \caption{The upper part is the origin graph $G$. The bottom part is an example in $\mathcal{G}_i$: for any $i \geq k$, any graph $(G(i);G(-i)) \in \mathcal{G}_i$, the graph $(G(i);G(-i))$ is generated by dividing $G$ into two parts. Then, $G(i)$ is generated by keeping the same as the first part, $G(-i)$ is then generated by deleting some of the edges in the second part. Note that there is no edge between $i$ and $i+1$.}
  \label{fig:upperbound}
\end{figure}

To get the upper bound of the approximation ratio, we consider a kind of ``worst-case" graphs where the contributions of agents except influential nodes can be ignored when $k \rightarrow \infty$. Since the mechanism $\mathcal{M}$ satisfies the fairness, it holds that $x_{i}(G') = x_{i}(G'')$ for any two graphs $G', G'' \in \mathcal{G}_i$. Then for any graph $G' \in \mathcal{G}_k$, agent $k$ is assigned the same probability. Thus, we can find that in the graph set $\mathcal{G}_k$, the ``worst-case" graph $G_k$ is a graph where there are only edges between $k$ and $i$, $1\leq i \leq k-1$ (shown in Figure~\ref{fig:worstcase}).

\begin{figure}[!htb]
  \centering
  \includegraphics[width=.6\linewidth]{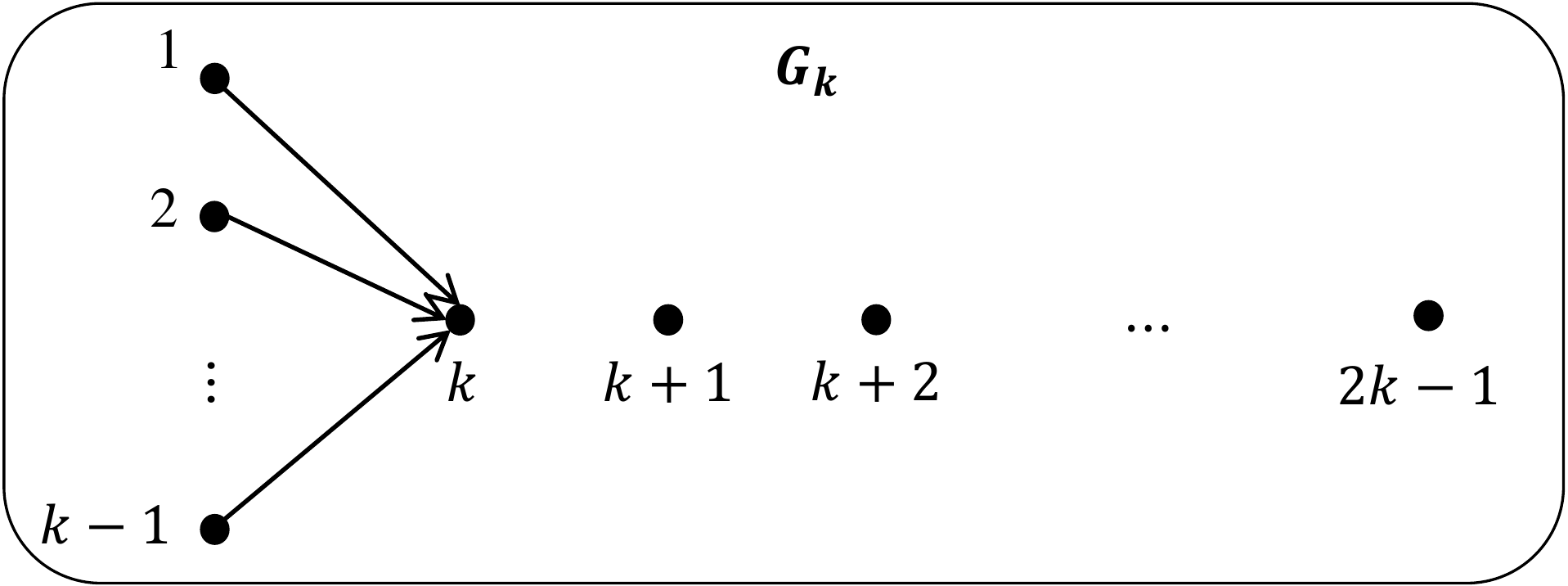}
  \caption{The ``worst-case" graph $G_k$ in the set $\mathcal{G}_k$.}
  \label{fig:worstcase}
\end{figure}
Since no matter how much the probability the mechanism assigns to other agents, the expected ratio for the graph $G_k$ approaches the probability $x_{k}(G_k)$ when $k \rightarrow \infty$, i.e., 
\[\lim_{k \rightarrow \infty} R(G_k) \leq \lim_{k \rightarrow \infty} x_{k}(G_k) + \frac{1}{k}\cdot (1-x_{k}(G_k)) = x_{k}(G_k).\]
The inequality holds since $\sum_{i=1}^{2k-1} x_{i}(G_k) \leq 1$. Similarly, for any $k<j \leq 2k-2$, the ``worst-case" graph $G_j$ in $\mathcal{G}_j$ is the graph where the out-edge of agent $i$ is deleted for all $i \geq j$. When $k \rightarrow \infty$, the expected ratio in the graph $G_j$ is
\[\lim_{k \rightarrow \infty} R(G_j) \leq \lim_{k \rightarrow \infty} \sum_{i=k}^j x_{i}(G_j)\cdot \frac{i}{j} + \frac{1}{j}\cdot \left(1-\sum_{i=k}^j x_{i}(G_j)\right) = \sum_{i=k}^j x_{i}(G_j)\cdot \frac{i}{j}.\]
Therefore, in these ``worst-case" graphs, we assume that only influential nodes can be assigned positive probabilities. Suppose that for the graph $G_j$, $k \leq j \leq 2k-2$, an influential node $i$ gets a probability of $x_{i} = \beta_{i-k}^{(j-i)}$ for $k \leq i \leq j$. 

Since the mechanism $\mathcal{M}$ is incentive compatible, it holds that $x_{i}(G) \geq x_{i}(G')$ for all $G' \in \mathcal{G}_i$ and all $i \in N$. To maximize the expected progeny of the selected agent in all graphs, we can set $x_{i}(G') = x_{i}(G)$ for all $G' \in \mathcal{G}_i$ and all $i \in N$. Similarly, it also holds that $x_{i}(G'') \geq x_{i}(G')$ for any $i \in N$, any $G' \in \mathcal{G}_i$, any $G'' \in \mathcal{G}_j$ and $k \leq i < j \leq 2k-1$. When $k \rightarrow \infty$, we can compute the performance of the mechanism $\mathcal{M}$ in different graphs as the following.
\begin{align*}
    &R(G_{j}) = \sum_{i=k}^{j} \beta_{i-k}^{(j-i)} \cdot \frac{i}{j}, k\leq j \leq 2k-2,\\
    &R(G) = \sum_{i=k}^{2k-1} \beta_{i-k} \cdot \frac{i}{2k-1}, 
\end{align*}
with $\beta_{i-k}^{(j-i)} \geq \beta_{i-k}^{(0)}$, $\beta_{i-k}^{(0)} = \beta_{i-k}$, $k\leq i \leq 2k-1$, $k\leq j \leq 2k-2$. The approximation ratio of the mechanism $\mathcal{M}$ should be at most the minimum of $R(G_j)$ and $R(G)$ for $k \leq j \leq 2k-2$, i.e.,
\begin{align*}
    r_{\mathcal{M}} \leq \min \left\{\beta_{0}^{(0)}, \beta_{0}^{(1)}\cdot \frac{k}{k+1} + \beta_{1}^{(0)}, \cdots, \beta_0\cdot \frac{k}{2k-1} + \beta_1\cdot \frac{k+1}{2k-1}+\cdots+\beta_{k-1}\right\}.
\end{align*}
Then we can choose $\beta_{i-k}^{(j-i)}$ to achieve the highest minimum expected ratio. We find that $r_{\mathcal{M}} \leq \frac{1}{1+\ln 2}$ and the equation holds when $k \rightarrow \infty$ and 
\begin{equation*}
    \begin{cases}
    \beta_{i-k}^{(j-i)} = \beta_{i-k}, \\
    \beta_0 + \beta_1 + \cdots + \beta_{k-1} = 1, \\
    \beta_0^{(0)} = \beta_0^{(1)}\cdot \frac{k}{k+1} + \beta_{1}^{(0)} = \cdots = \beta_0 \cdot \frac{k}{2k-1} + \beta_1\cdot \frac{k+1}{2k-1}+\cdots+\beta_{k-1}.
    \end{cases}
\end{equation*}
\end{proof}

\subsection{Open Questions}\label{openproblems}
Note that the approximation ratio of the Geometric Mechanism is close to the upper bound we prove in Section~\ref{ub}. However, there is still a gap between them. In this section, we suggest two open questions which narrow down the space for finding the optimal selection mechanism.

\subsubsection{Root Mechanism.}
Recall that our goal in this paper is to maximize the approximation ratio between the expected progeny of the selected agent and the maximum progeny. If requiring incentive compatibility, a selection mechanism cannot simply select the most influential node. However, we can identify a subset of agents who can pretend to be the most influential node. This is the influential set we illustrate in Definition~\ref{def:infset}, and we show that agents cannot be placed into the influential set by misreporting as illustrated in Observation~\ref{ob:set}. Utilizing this idea, we see that if we assign positive probabilities only to these agents, then the selected agent has a large progeny, and agents have less chance to manipulate. We call such mechanisms as \emph{root mechanisms}.

\begin{definition}
A root mechanism $\mathcal{M}_R$ is a family of functions $f_R: \Theta^n \rightarrow [0,1]^{n}$ for all $n \in \mathbb{N}^*$. Given a set of agents $N$ and their report profile $\theta'$, a root mechanism $\mathcal{M}_R$ only assigns positive selection probabilities to agents in the set $S^{inf.}(G(\theta'))$. Let $x_i(\theta')$ be the probability of selecting agent $i\in N$. Then $x_i(\theta') = 0$ for all $i \notin S^{inf.}(G(\theta'))$,  $x_i(\theta') \in [0,1]$ for all $i \in N$ and $\sum_{i\in N} x_i(\theta') \leq 1$.
\end{definition}
It is clear that our Geometric Mechanism is a root mechanism, whose approximation ratio is not far from the upper bound of $1/(1+\ln 2)$. We conjecture that an optimal incentive compatible selection mechanism and an optimal incentive compatible root mechanism share the same approximation ratio bound.
\begin{conjecture}
If an optimal incentive compatible root mechanism $\mathcal{M}_R$ has an approximation ratio of $r_{\mathcal{M}_R}^*$, there does not exist other incentive compatible selection mechanism that can achieve a strictly better approximation ratio. 
\end{conjecture}\label{up root}
\begin{proof}[Discussion]
An optimal incentive compatible selection mechanism will usually try to assign more probabilities to agents with more progeny. Following this way, we assign zero probability to all agents who are not an influential node and find a proper probability distribution for the influential set, rather than giving non-zero probabilities to all agents. Since any agent who is not an influential node cannot make herself in the influential set when other agents' reports are fixed, this method will not cause a failure for incentive compatibility.
\end{proof}
\subsubsection{Fairness.}\label{consistency}
Note that the upper bound of $1/(1+\ln 2)$ is for all incentive compatible and fair selection mechanisms. We should also consider whether an incentive compatible selection mechanism can achieve a better approximation ratio without the constraint of fairness. Here, we conjecture that an incentive compatible selection cannot achieve an approximation ratio higher than $1/(1+\ln 2)$ if the requirement of fairness is relaxed. 

\begin{conjecture}
If an optimal incentive compatible and fair mechanism $\mathcal{M}$ can achieve an approximation ratio of $r_{\mathcal{M}}^*$, there does not exist other incentive compatible mechanism with a strictly higher approximation ratio.
\end{conjecture}

\begin{proof}[Discussion]
Let $\mathcal{G}_f$ be a set of graphs where for any two graphs $G,G' \in \mathcal{G}_f$, their number of nodes, their influential sets $S^{inf.}(G) = S^{inf.}(G') = \{s_1,\cdots,s_m\}$ and the subgraphs constructed by agent $s_1$'s progeny are same. If an incentive compatible selection mechanism is not fair, there must exist such a set $\mathcal{G}_f$ where the mechanism fails fairness. Then the expected ratios in two graphs in $\mathcal{G}_f$ may be different, and the graph with a lower expected ratio might be improved since these two graphs are almost equivalent. One possible way for proving this conjecture is to design a function that reassigns probabilities for all graphs in $\mathcal{G}_f$ such that $x_{s_1}$ is the same for these graphs without hurting the property of incentive compatibility, and all graphs in $\mathcal{G}_f$ then share the same expected ratio without hurting the efficiency of the selection mechanism.
\end{proof}

\section{Conclusion}
In this paper, we investigate a selection mechanism for choosing the most influential agent in a network. We use the progeny of an agent to measure her influential level so that there exist some cases where an agent can decrease her out-edges to make her the most influential agent. We target selection mechanisms that can prevent such manipulations and select an agent with her progeny as large as possible. For this purpose, we propose the Geometric Mechanism that achieves at least $1/2$ of the optimal progeny. We also show that no mechanism can achieve an expected progeny of the selected agent that is greater than $1/(1+\ln 2)$ of the optimal under the conditions of incentive compatibility and fairness. 

There are several interesting aspects that have not been covered in this paper. First of all, there is still a gap between the efficiency of our proposed mechanism and the given upper bound. One of the future work is to find the optimal mechanism if it exists. In this direction, we also leave two open questions for further investigations. Moreover, selecting a set of influential agents rather than a single agent is also important in real-world applications (e.g., ranking or promotion). So another future work is to extend our results to the settings where a set of $k$ ($k>1$) agents need to be selected.%
%
%
\bibliographystyle{splncs04}
\bibliography{bib}
\end{document}